\documentclass[10pt, conference, letterpaper]{IEEEtran}
\IEEEoverridecommandlockouts
\usepackage{graphicx}
\usepackage{epstopdf}
\usepackage{cite}
\usepackage{amsmath,amssymb,amsfonts}
\usepackage{amsmath} 
\usepackage{amssymb}
\usepackage{algorithmic}
\usepackage{graphicx}
\usepackage{textcomp}
\usepackage{pifont}
\usepackage{xcolor}
\usepackage{enumitem,kantlipsum}
\usepackage{comment}
\usepackage{url}
\usepackage{amsthm}
\usepackage{balance}

\usepackage{silence}
\usepackage{caption}
\usepackage{subcaption}
\usepackage{multirow}
\usepackage{setspace}

\usepackage{soul}

\usepackage{mathtools}
\usepackage{booktabs}

\newtheorem{thm}{Theorem}
\newtheorem{lemma}{Lemma}

\theoremstyle{remark}

\theoremstyle{definition}

\def\name{\textit{Prism}\xspace} 

\usepackage{accents}

\usepackage{mathtools}
\usepackage{bbm}
\usepackage{comment}
\newcommand{\para}[1]{\noindent {\bf #1}}
\usepackage[linesnumbered,ruled,vlined]{algorithm2e}

\usepackage[tr]{optional}

\MakeRobust{\underaccent}
\def\BibTeX{{\rm B\kern-.05em{\sc i\kern-.025em b}\kern-.08em
		T\kern-.1667em\lower.7ex\hbox{E}\kern-.125emX}}

\begin{document}

\title{Accelerating Edge Inference for Distributed MoE Models with Latency-Optimized Expert Placement}

\author{\IEEEauthorblockN{
        Tian Wu\IEEEauthorrefmark{1},
        Liming Wang\IEEEauthorrefmark{2},
        Zijian Wen\IEEEauthorrefmark{1},
        Xiaoxi Zhang\IEEEauthorrefmark{1}\IEEEauthorrefmark{5},
        Xu Chen\IEEEauthorrefmark{1},\\
        Jingpu Duan\IEEEauthorrefmark{3},
        Xianwei Zhang\IEEEauthorrefmark{1},
        Jinhang Zuo\IEEEauthorrefmark{2}
        }
        \IEEEauthorblockA{
        \IEEEauthorrefmark{1}Sun Yat-sen University
        \IEEEauthorrefmark{2}City University of Hong Kong
        \IEEEauthorrefmark{3}Peng Cheng Laboratory
        }
        \IEEEauthorblockA{
        Email: 
        \IEEEauthorrefmark{1}\{wutian,wenzj7\}@mail2.sysu.edu.cn,
        \IEEEauthorrefmark{2}w.lm@my.cityu.edu.hk,
        \IEEEauthorrefmark{3}duanjp@pcl.ac.cn
        }
        \IEEEauthorblockA{
        \IEEEauthorrefmark{2}jinhang.zuo@cityu.edu.hk,
        \IEEEauthorrefmark{1}\{zhangxx89,chenxu35,zhangxw79\}@mail.sysu.edu.cn
        }
        \thanks{\IEEEauthorrefmark{5}Corresponding author: Xiaoxi Zhang.}
}

    \IEEEoverridecommandlockouts
	\maketitle
	\IEEEpubidadjcol
	
	\begin{abstract}
        The emergence of Mixture-of-Experts (MoE) has transformed the scaling of large language models by enabling vast model capacity through sparse activation. Yet, converting these performance gains into practical edge deployment remains difficult, as the massive memory footprint and communication demands often overwhelm resource-limited environments. While centralized cloud-based solutions are available, they are frequently plagued by prohibitive infrastructure costs, latency issues, and privacy concerns. Moreover, existing edge-oriented optimizations largely overlook the complexities of heterogeneous hardware, focusing instead on isolated or uniform device setups.
        In response, this paper proposes \name, an inference framework engineered for collaborative MoE serving across diverse GPU-equipped edge servers. By leveraging the intrinsic sparsity and input locality of MoE workloads, \name minimizes inter-server communication and optimizes expert placement within diverse resource constraints. The framework integrates an activation-aware placement strategy that balances local request coverage with memory utilization, supplemented by a runtime migration mechanism to adapt expert distribution to dynamic workload changes.
        Experiments on contemporary MoE models and datasets demonstrate that \name reduces inference latency by up to 30.6\% and significantly lowers communication costs compared to state-of-the-art baselines, confirming the effectiveness of cooperative edge-based MoE serving.
	\end{abstract}

	\section{Introduction}

Mixture-of-Experts (MoE) architectures have emerged as a central design choice for training large language models (LLMs) at scale. By integrating multiple specialized subnetworks, known as experts, MoE achieves substantial performance gains without a proportional increase in computational cost. A lightweight gating mechanism dynamically routes each input token to only a small subset of experts, enabling sparse activation and efficient parallelism during training. This sparsity enables MoE to scale significantly while keeping training costs manageable, leading to its adoption in state-of-the-art models such as Switch Transformer~\cite{fedus2022switch}, Mixtral~\cite{jiang2024mixtral}, and DeepSeek-V3~\cite{liu2024deepseek}. As a result, MoE has become a foundational architecture for training and scaling modern LLMs.

\begin{figure}[htb]
    \centering
     \includegraphics[width=0.48\textwidth]{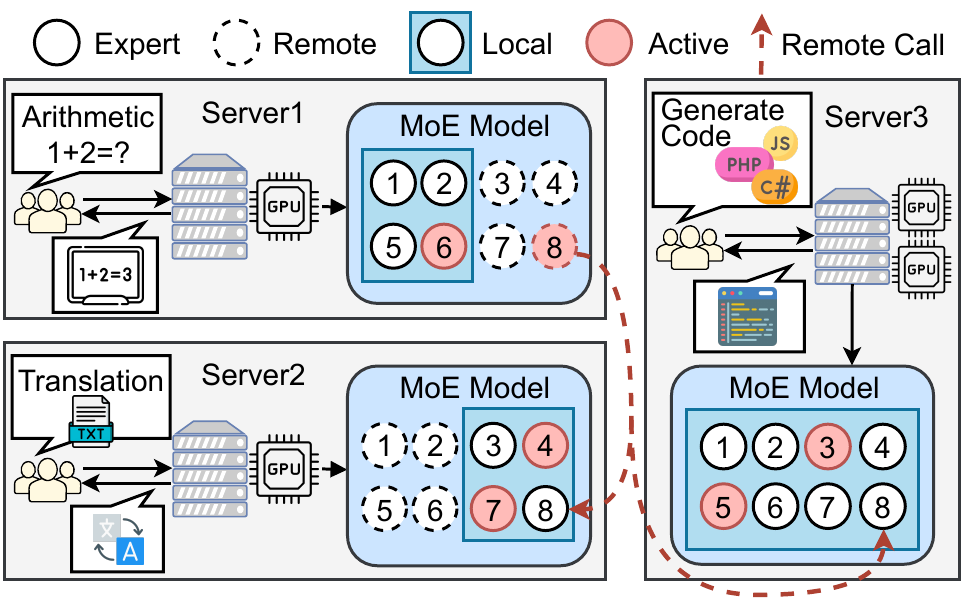}
    \caption{Illustration of distributed MoE inference across three edge servers. Each server handles requests from its own users and hosts a subset of experts.  When a required expert is not available locally, the server performs remote computation by invoking the expert on another node.} 
    \label{fig:moe_collaboration}
\end{figure}

Despite its efficiency during training, MoE inference remains resource-intensive, especially in terms of GPU memory, due to growing sizes of modern MoE models. Even with sparse activation, the cumulative memory footprint often exceeds the capacity of a single GPU. For example, Mixtral-8×7B~\cite{jiang2024mixtral} requires over 80 GB of GPU memory, far exceeding the 16--24 GB typically available on edge servers with commodity GPUs, such as those equipped with NVDIA RTX 4090 or A4000 GPUs commonly found in university labs, hospitals, small and medium-sized enterprises~\cite{nvdia4090,gpuion}. As a result, the standard practice is to offload inference to cloud servers with large-scale GPU clusters. However, relying solely on centralized cloud services presents many limitations: prohibitively high expenditures of renting cloud GPU instances, high response latency from network round-trips, and growing concerns over data privacy and regulatory compliance. 

\para{Why collaborative edge inference for MoE?} In response, distributed deployments of LLMs including MoE models on GPU-equipped edge servers has emerged as a compelling alternative. Many organizations (e.g., universities and research labs) coordinate underutilized lab servers or desktops with modest GPU capabilities in different locations~\cite{gpuion}. By doing so, edge-based inference offers significantly lower infrastructure costs and improved data locality.
However, these methods still incur high MoE inference latency, as they lack fine-grained analysis of MoE model architectures and expert activation patterns. A few recent works have proposed strategies for edge MoE, but they either target acceleration on a single edge device~\cite{edgemoe} or overlook the system and request heterogeneity in edge scenarios~\cite{fast-edge-moe,go2025moetuner}: {\em resource configurations, request volumes, and expert usages can all vary across distributed edge servers.}
To address this, we propose a collaborative MoE inference architecture, inspired by Expert Parallelism (EP), which distributes each layer of experts across GPUs. As illustrated in Fig.~\ref{fig:moe_collaboration}, inference is performed collaboratively across multiple devices, with each server hosting only a subset of the experts due to memory constraints. When a request arrives at Server 1 that activates an expert not present locally, such as Expert 8, the computation can still be completed by invoking the expert remotely from another server like Server 2 or Server 3. This setup leverages idle edge resources to form a cooperative inference layer, mitigating memory bottlenecks without relying on cloud execution.

\para{Importance and challenges of expert placement.} 
Despite its advantages, collaborative MoE edge inference still faces challenges, particularly inter-GPU/server communication latency. 
Existing distributed MoE inference methods~\cite{shoeybi2019megatron,go2025moetuner,deepseek2025eplb} including EP solutions designed for datacenters are unsuitable in edge scenarios due to two reasons: (i) edge servers vary in GPU memory, compute capacity, and network reliability, unlike datacenters with uniform hardware and high-speed interconnects; and (ii) redirecting requests across locations is costly, making it essential to process in proximity to data, a core principle of edge computing.
These constraints call for strategic expert placement to reduce inference latency by maximizing local execution and minimizing remote calls. 
Additionally, expert activation patterns vary across servers due to user behavior and workload differences. While stable over short periods, longer-term shifts can degrade performance if placements remain fixed. Lightweight, periodical placement adjustment is needed to adapt layouts as workloads evolve, without complex migration protocols or frequent re-optimization.

\para{Insights.} Fortunately, the inherent sparsity and modularity of MoE models offer an opportunity to address these challenges. Since only a small subset of experts is activated per input, systems can be designed to favor local execution. As explored in Section~\ref{sec:Activation_Patterns}, servers often exhibit distinct patterns of expert activation due to varying user tasks and input distributions, allowing each node to prioritize caching its most frequently accessed experts. As shown in Fig.~\ref{fig:moe_collaboration}, Server 1 and Server 2 each store only four out of eight experts, strategically selecting those most relevant to their local workloads. Consequently, Server 2 can process inputs routed to Experts 4 and 7 entirely locally, while Server 1 incurs occasional cross-server communication. Crucially, such remote calls remain infrequent, keeping end-to-end latency within acceptable bounds.
Leveraging these insights, we propose \name, a novel Distributed activation-aware collaborative MoE inference system tailored to resource-constrained edge environments. Our approach explicitly addresses the challenges of memory limitations, communication overhead, and workload heterogeneity, while capitalizing on the sparse activation patterns intrinsic to MoE architectures. The key contributions of this work are:

\begin{enumerate}[wide, labelindent=0pt]
\item \emph{System Design.} We introduce \name, a novel distributed inference system for collaborative MoE serving across memory-constrained, heterogeneous edge servers. It exploits activation sparsity and local workload patterns to minimize cross-server communication and enable scalable deployment.

\item \emph{Activation-aware Expert Placement.} We propose a data-driven placement algorithm that allocates experts for different requests across heterogeneous servers based on activation frequencies. Our design balances local coverage and memory usage, and offers theoretical approximation guarantees.

\item \emph{Lightweight Expert Migration.} To adapt to evolving workloads and input distributions, \name periodically re-evaluates placement and triggers efficient migrations when beneficial—without incurring high coordination overhead.

\item \emph{Comprehensive Experiments.} We evaluate \name using multiple MoE models and real-world data. Compared to existing approaches, \name reduces cross-server communication and achieves up to 30.6\% lower inference latency.
\end{enumerate}

        \section{Motivation}
\label{sec:pre}

\subsection{Expert Activation Patterns}
\label{sec:Activation_Patterns}
A core property of MoE models is their sparse activation: only a small subset of experts is activated per layer for each input. While expert activations may appear evenly distributed across general-purpose tasks, we make a key observation: for \textbf{specific task types}, the activation patterns become highly skewed and task-dependent.
As illustrated in Fig.~\ref{fig:layer_0_comparison}, when the Mixtral~\cite{jiang2024mixtral} model processes arithmetic tasks versus ASCII word recognition tasks from BIG-bench dataset~\cite{srivastava2023beyond}, the expert utilization at the first layer varies significantly. Arithmetic tasks exhibit a strongly skewed activation pattern where Expert 0 dominates, while Expert 1 is rarely used. In contrast, ASCII word recognition predominantly activates Expert 3.

This observation motivates a straightforward yet effective strategy: if an edge server primarily handles a specific task type, it should prioritize loading \textit{high-frequency experts} into local GPU memory to enable faster inference. The remaining low-frequency experts can be handled via memory offloading or remote cooperation with other devices (discussed in the next subsection). In short, {\em deploying frequently used experts locally while deferring infrequent ones offers a promising direction for resource-efficient MoE inference.}

\begin{figure}[htb]
\centering
\includegraphics[width=0.40\textwidth]{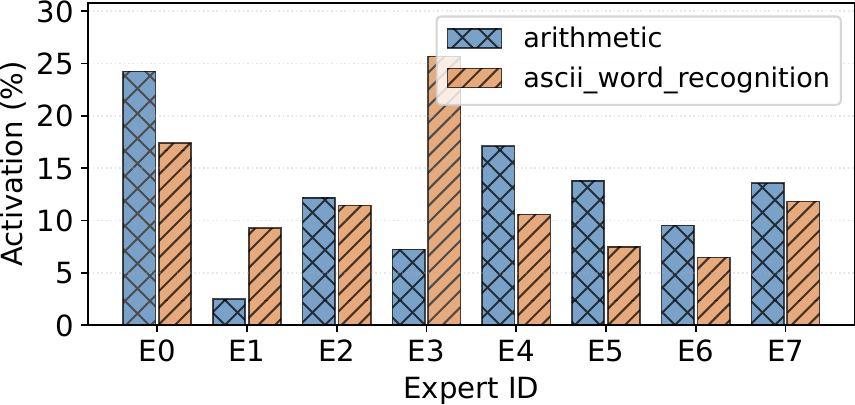}
\caption{Activation patterns across tasks.}
\label{fig:layer_0_comparison}
\end{figure}

\begin{figure}[htb]
\centering
\includegraphics[width=0.42\textwidth]{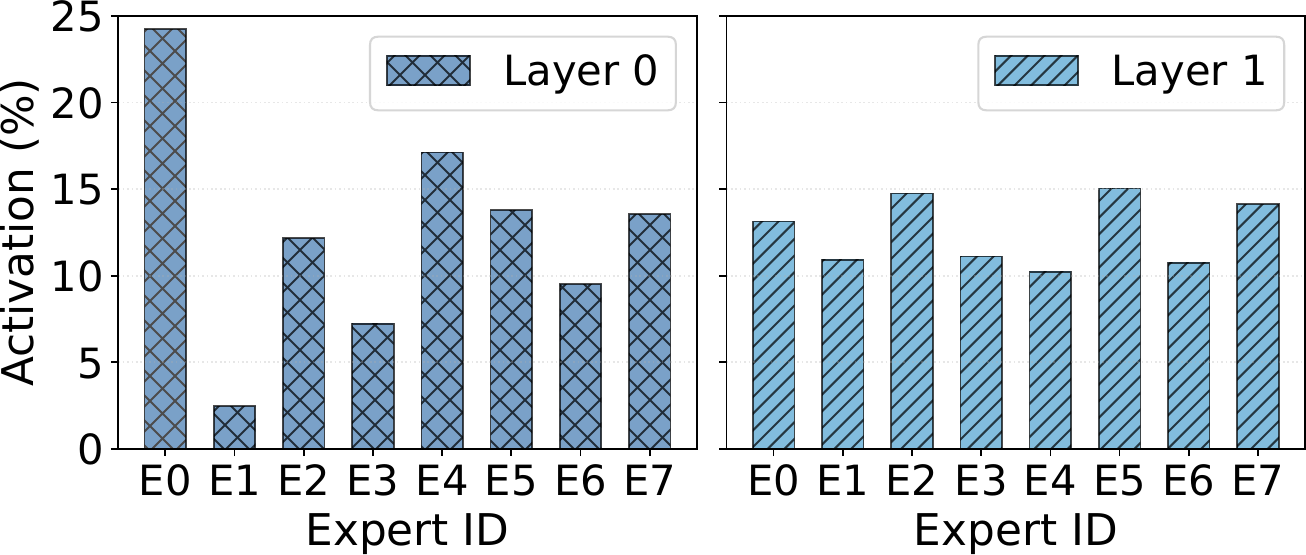}
\caption{Activation patterns across layers.}
\label{fig:layer_arithmetic_cdf}
\end{figure}

However, a further challenge emerges: expert activation patterns can also {\em vary across layers}, even within the same task. As shown in Fig.~\ref{fig:layer_arithmetic_cdf}, Layer 0 of the arithmetic task shows a highly skewed pattern (e.g., Expert 0 dominates), while Layer 1 demonstrates a more uniform distribution.
This layer-wise variation introduces a second dimension to the expert placement problem: {\em how to allocate limited GPU memory across layers}. Layer 0 could be handled locally using a smaller expert subset, while Layer 1 may require more memory to accommodate its broader activation footprint.

In summary, while activation patterns present valuable opportunities for optimizing distributed MoE inference, effective exploitation requires joint consideration of task-dependent and layer-wise variations. Our system design takes both dimensions into account when devising expert placement strategies. However, memory-local strategies alone are insufficient under system heterogeneity and dynamic workloads. To address these limitations, we now explore the potential of expert cooperation across edge nodes.

\subsection{The Need of Collaborative Inference}
As discussed above, one way to accommodate low-frequency experts is through memory offloading. For instance, MoE-Infinity~\cite{moe-infinity} enables single-device MoE inference by dynamically loading rarely used experts from RAM to GPU memory, while frequently activated experts remain cached on the GPU. However, such methods are limited to a single-machine setup and cannot take advantage of the collective computing power available in edge environments. In practice, edge servers often exhibit heterogeneous workloads and resource capacities. For example, Mooncake’s open-source online trace dataset~\cite{qin2025mooncake} shows that ToolAgent requests occur much more frequently than conversational queries. Moreover, the hardware heterogeneity among servers, such as differences in GPU availability or compute capability, leads to resource imbalance: {\em some nodes remain underutilized while others are bottlenecked}. To address this, a natural idea is to redistribute the load across servers.

One straightforward solution is to redirect requests to idle servers, i.e., reroute incoming requests to another server that performs inference locally using memory offloading. While this improves load balancing to some extent, our experiments show that it remains suboptimal.
To highlight this, we compare three methods in Table~\ref{tab:motivation_inf}. 
The experiment uses the Mixtral-8×7B model deployed across three simulated edge servers. Each server processes a distinct type of request drawn from three datasets: arithmetic reasoning, ASCII word recognition, and abstract narrative understanding (from BIG-bench~\cite{srivastava2023beyond}).
MoE-Infinity (w/ LB) represents the request-redirection baseline. In contrast, Naive Collaboration deploys experts randomly across the servers and enables distributed inference with remote expert calls. Despite its simplicity, the collaborative setup achieves significantly lower average latency, showing more balanced load across all servers.

These findings underscore the potential of collaborative MoE inference, where expert modules are distributed across devices and invoked jointly during inference. Even without fine-tuned optimization, this architecture improves load balancing and resource utilization. Motivated by this observation, our work goes further by proposing an activation-aware collaborative framework tailored for distributed MoE inference under system heterogeneity and memory constraints.

\begin{table}[htbp]
\centering
\caption{Average inference latency across different methods.}
\label{tab:motivation_inf}
\resizebox{0.49\textwidth}{!}{
\begin{tabular}{lcccc}
\hline
\textbf{Method} & \textbf{Server 1} & \textbf{Server 2} & \textbf{Server 3} & \textbf{Total Avg} \\
\hline
MoE-Infinity & 9.14 & {4.77} & {3.18} & 5.19 \\
MoE-Infinity (w/ LB) & {8.60} & {4.77} & \textbf{3.10} & {5.03} \\
Naive Collaboration & \textbf{4.96} & \textbf{4.34} & 3.40 & \textbf{4.11} \\
\hline
\end{tabular}
}
\end{table}
        \section{System Design}
\label{sec:sys}
\begin{figure*}[!t]
    \centering
    \includegraphics[width=0.9\textwidth]{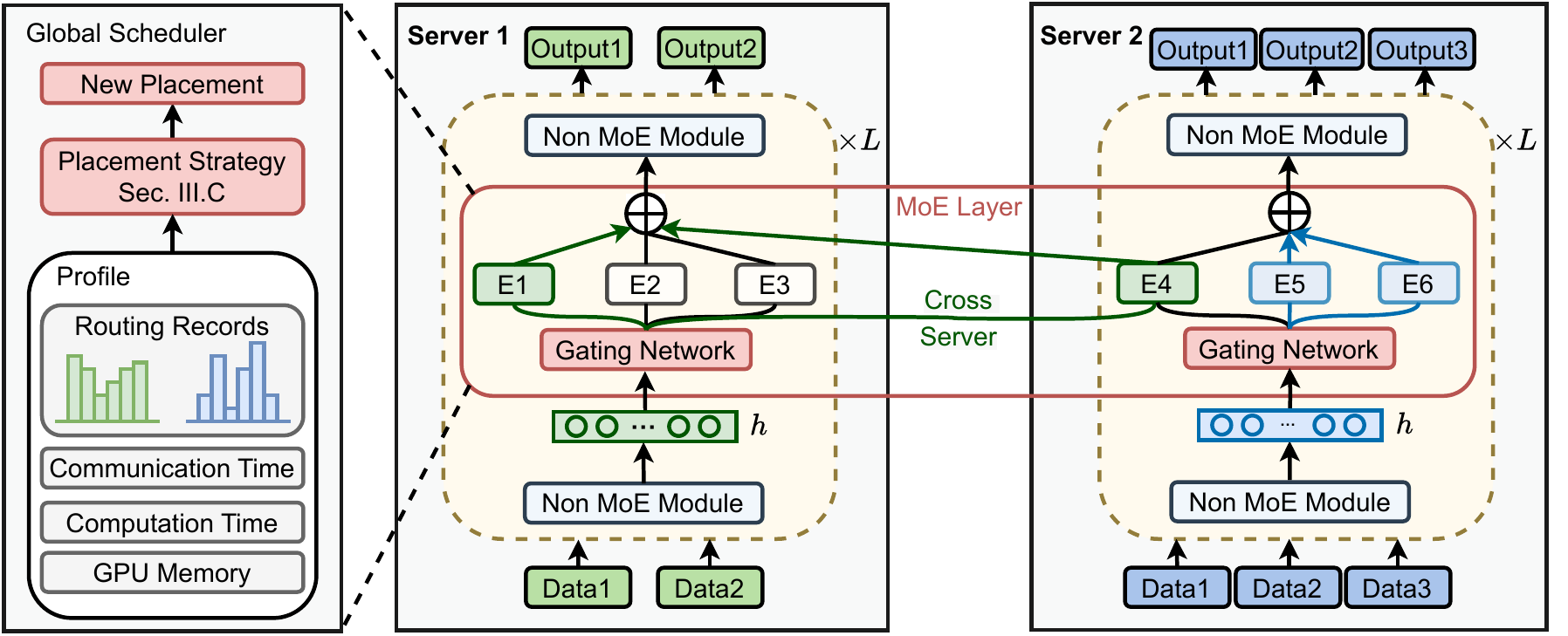}
    \caption{The workflow of \name. The system consists of two primary components working in coordination to enable efficient distributed inference for MoE models: a global scheduler and a runtime multi-server system that executes inference.}
    \label{fig:framework}
\end{figure*}    
\subsection{System Architecture}
\label{sec:arch}
As shown in Fig.~\ref{fig:framework}, our system adopts a distributed execution model coordinated by a global scheduler, which enables efficient distributed inference for MoE models.

On the left, the Global Scheduler acts as the central intelligence of the system, maintaining a comprehensive view of the distributed environment. It continuously collects and updates system-wide profiles, including cross-server communication latencies, computational capabilities, GPU memory, and historical patterns of expert activation across servers. These insights drive the Placement Strategy module, which dynamically determines the optimal allocation of experts to servers based on current resource utilization and historical workload. This adaptive placement mechanism minimizes cross-server communication and maximizes local hit rates for frequently accessed experts, with algorithmic details elaborated in Section~\ref{sec:our_placement}.
The right side of the figure shows that Server 1 and Server 2 host a subset of MoE experts (e.g., E1–E3 on Server 1, E4–E6 on Server 2). Each server is equipped with a local Gating Network that, for each input token, determines the most suitable experts to activate, i.e., {\em whether they are located locally or on remote servers}.
Input data (e.g., Data1) flows through the system, first passing through  non-MoE modules to compute the hidden state $h$. At the MoE layer, the Gating Network routes $h$ to the selected experts. In the example shown, Data1 arriving at Server 1 is routed to E1 (resident on Server 1) and E4 (hosted on Server 2). E1’s computation is local, while a remote call is initiated to Server 2 to execute E4. The results of experts are transmitted back to Server 1, where they are aggregated and passed through subsequent non-MoE layers and additional $L-1$ blocks to generate the output. 

Crucially, each inference request feeds system observability: gating decisions, expert invocation costs are logged and reported to the Global Scheduler. This enables continuous monitoring of workload and resource usage. Periodically, the scheduler analyzes the collected data to refine expert placement, migrating experts in response to shifting access patterns and maintaining efficiency in evolving edge environments.

\subsection{Problem Formulation}
\label{sec:model}
We then formulate the \textit{expert placement problem} for collaborative MoE inference into a constrained optimization problem. For theoretical clarity, we begin by assuming that the expert activation patterns for all input batches are known a priori. Under this assumption, the goal is to determine an optimal placement strategy that minimizes the total end-to-end latency when serving inference requests across a set of heterogeneous edge servers. While this assumption does not hold in practice, it provides a principled foundation for understanding the optimal structure of expert placement. In real-world systems, such activation statistics can be estimated from historical data or initialized randomly. The placement can then be refined online using migration strategies, which we discuss in Section~\ref{sec:our_placement}.

Consider a system consisting of $N$ edge servers, where each server $n \in [N]$ is equipped with $G_n$ GPUs. The MoE model contains $L$ layers, each with a set of experts denoted by $\mathcal{E}_l$ for layer $l \in [L]$. In practice, we have $|\mathcal{E}_l| = E$, for each layer $l$. Let $x_n^t$ denote the $t$-th input batch received by server $n$ and $X_n$ denote the set of all input batches assigned to $n$. An expert placement strategy is denoted by $\mathcal{P} = \{z_{n,g}^e\}$, where $z_{n,g}^e \in \{0,1\}$ indicates whether expert $e \in \cup_{l=1}^L \mathcal{E}_l$ is ($z_{n,g}^e=1$) placed on GPU $g$ of server $n$, or not ($z_{n,g}^e=0$). The objective is to minimize the total end-to-end inference latency across all input batches and servers, formalized as:
\begin{align}
\label{eq:min_cost}
\min_{\mathcal{P}} \quad \sum_{n = 1}^N \sum_{x_n^t \in X_n} \sum_{l = 1}^L T(x_n^t, l, \mathcal{P}).
\end{align}
Here, $T(x_n^t, l, \mathcal{P})$ denotes the processing latency at layer $l$ for input $x_n^t$ under placement strategy $\mathcal{P}$:
\begin{align*}
T(x_n^t, l, \mathcal{P}) = \max_{\hat{n}, g, e} \left( T_{\text{comm}}(y_{\hat{n},g}^e(x_n^t, l)) + T_{\text{comp}}(y_{\hat{n},g}^e(x_n^t, l)) \right)
\end{align*}
Here, $y_{\hat{n},g}^e(x_n^t, l)$ denotes the intermediate output from layer $l$ of input $x_n^t$ that is routed to expert $e$ on GPU $g$ of server $\hat{n}$. The terms $T_{\text{comm}}$ and $T_{\text{comp}}$ capture the communication and computation latencies, respectively. The outer maximization accounts for the fact that all expert outputs at a given layer must be aggregated before proceeding to the next layer; thus, the slowest invocation dominates the layer’s processing time.

However, in practice, this latency objective is difficult to optimize directly. Latency prediction is inherently noisy due to queuing dynamics, input variability, and hardware-level contention. To address this, we propose a tractable proxy objective that captures the expected volume of cross-server expert invocations, the key source of latency in distributed MoE inference.
This reformulation is motivated by our empirical observations: as shown in Fig.~\ref{fig:layer_latency}, the inference latency per layer increases sharply with the proportion of experts executed remotely. This degradation is primarily due to multi-stage communication overhead: transmitting activations over the network to a remote server’s RAM, loading them into GPU memory, executing the expert, and sending results back to the original server. This sequence introduces a significant performance bottleneck.

\begin{figure}[htb]
    \centering
    \includegraphics[width=0.3\textwidth]{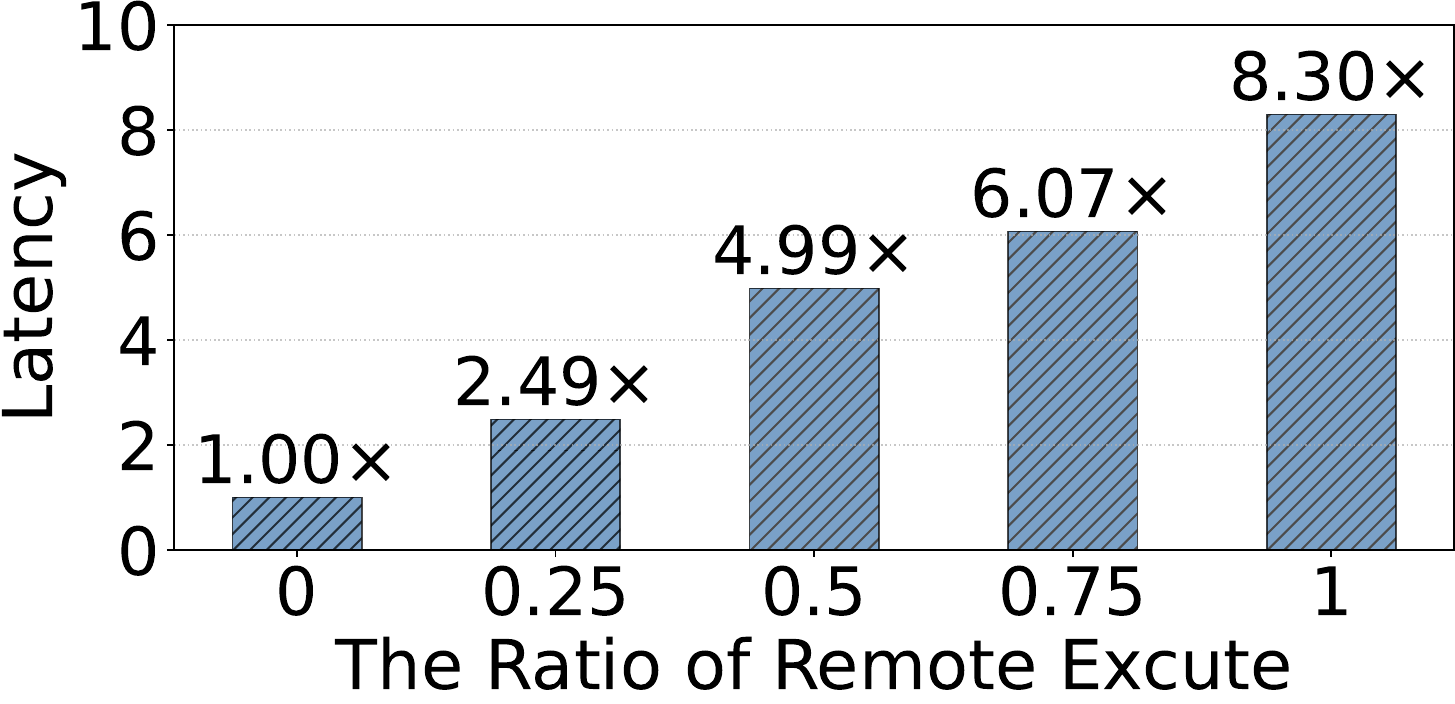}
    \caption{Layer-wise inference latency increases with the proportion of experts executed on remote servers.}
    \label{fig:layer_latency}
\end{figure}

Let $f_n^l(e)$ denote the empirical activation frequency of expert $e \in \mathcal{E}_l$ on server $n$ at layer $l$. To determine whether a remote invocation occurs, we define the indicator:
\begin{align*}
\mathbf{1}_{\text{remote}}(n,e) =
\begin{cases}
1, & \text{if } \sum_{g=1}^{G_n} z_{n,g}^e = 0 \\
0, & \text{otherwise}
\end{cases}
\end{align*}
That is, server $n$ must invoke expert $e$ remotely if the expert is not locally available on any of its GPUs. Using this indicator, we define the communication-aware proxy objective:
\begin{align}
\min_{\mathcal{P}} \quad \sum_{n = 1}^N \sum_{l = 1}^L \sum_{e \in \mathcal{E}_l} f_n^l(e) \cdot \mathbf{1}_{\text{remote}}(n,e)
\label{eq:comm_aware_obj}
\end{align}
This expression captures the expected number of remote expert invocations, weighted by how frequently each expert is used in each layer. Minimizing this proxy objective provides a practical and effective surrogate for reducing overall inference latency.
The optimization is subject to two constraints:
\begin{align*}
& \sum_{n = 1}^N \sum_{g = 1}^{G_n} z_{n,g}^e \geq 1, \quad \forall e \in \bigcup_{l=1}^L \mathcal{E}_l \\
& \sum_{l = 1}^L \sum_{e \in \mathcal{E}_l} z_{n,g}^e \cdot m_e \leq \texttt{mem}_{n,g}, \quad \forall n, g
\end{align*}
The first constraint ensures that every expert is deployed on at least one GPU in the system (expert coverage). The second enforces per-GPU memory limits, where $m_e$ is the memory footprint of expert $e$ and $\texttt{mem}_{n,g}$ denotes the available memory on GPU $g$ of server $n$.
Together, this formulation captures the key tradeoffs of distributed MoE inference: satisfying expert coverage and hardware constraints while minimizing cross-server communication. It provides a solid theoretical foundation for developing activation-aware and memory-efficient expert placement strategies.

\subsection{Activation-Aware Expert Placement}
\label{sec:our_placement}

We propose a two-stage algorithm for expert placement that respects memory and expert-coverage constraints while reducing remote invocations. The first stage (Algorithm~\ref{alg:expert_count}) determines how many experts from each MoE layer should be placed on each server. The second stage (Algorithm~\ref{alg:expert_assign}) selects which specific experts to assign, based on activation patterns.

\subsubsection{\textbf{Layer-wise Expert Count Allocation}}
Algorithm~\ref{alg:expert_count} distributes the total number of experts per layer across all servers by balancing activation diversity and memory capacity. It begins by allocating the expert count $N_{n,l}$ for each server $n$ and layer $l$ in proportion to the local activation diversity observed at that server. Specifically, we compute the normalized entropy of expert activations in layer $l$ on server $n$ as:
\[
v_{n,l} = -\sum_{e \in \mathcal{E}_l} p_e \log_2 p_e
\quad \text{where } p_e = \frac{f_n^l(e)}{\sum_{e' \in \mathcal{E}_l} f_n^l(e')}.
\]
This Shannon entropy captures the variety of expert usage: it is low when a few experts dominate (low diversity) and high when activations are spread evenly. Then, each expert count is initialized as $N_{n,l} = \left\lfloor \frac{M_n}{m_e} \cdot \frac{v_{n,l}}{\sum_{l=1}^L v_{n,l}} \right\rfloor$, with server memory defined as $M_n = \sum_g \texttt{mem}_{n,g}$. This heuristic reflects the intuition that more diverse activation patterns require a broader set of experts to be co-located.
To satisfy the expert-coverage constraint from the problem formulation, i.e., $\sum_n N_{n,l} \geq E$ for each layer, we adjust the initial allocation through expert rebalancing. If the total expert count in a layer $l$ is below $E$, we iteratively borrow experts from layers $l'$ that are currently over-provisioned. Servers are prioritized in descending order of memory capacity to ensure feasibility. This process repeats until all layer-wise totals match their required expert counts.

\begin{algorithm}[t]
\caption{Layer-wise Expert Count Allocation}
\label{alg:expert_count}
\SetAlgoVlined
\SetKwInOut{Input}{Input}
\SetKwInOut{Output}{Output}

\Input{Total experts per layer $E$, expert size $m_e$, \\
server memory $M_n = \sum_g \texttt{mem}_{n,g}$, \\
activation entropy $v_{n,l}$ for server $n$ layer $l$}
\Output{Expert count $N_{n,l}$ for server $n$ and layer $l$}

\vspace{1mm}
// Step 1: Initialization based on activation diversity.\\
\For{$n=1$ to $N$}{
    \For{$l=1$ to $L$}{
        $N_{n,l} \gets \left\lfloor \frac{M_n}{m_e} \cdot \frac{v_{n,l}}{\sum_{l=1}^L v_{n,l}} \right\rfloor$\;
    }
}

// Step 2: Adjust to match expert-coverage constraint.\\
\For{$l=1$ to $L$}{
    $N_{l,\text{total}} = \sum_{n=1}^N N_{n,l}$\;
}
\For{$l=1$ to $L$}{
    \While{$N_{l,\text{total}} < E$}{
        \For{$n=1$ to $N$ (sorted by $M_n$ descending)}{
            $l' = \arg\max_{l'} N_{l',\text{total}}$\;
            \If{$N_{n,l'} > 0$}{
                $N_{n,l'} \gets N_{n,l'} - 1$\;
                $N_{n,l} \gets N_{n,l} + 1$\;
                $N_{l',\text{total}} \gets N_{l',\text{total}} - 1$\;
                $N_{l,\text{total}} \gets N_{l,\text{total}} + 1$\;
                \If{$N_{l,\text{total}} == E$}{
                    \textbf{break}\;
                }
            }
        }
    }
}
\end{algorithm}

\textit{Justification of Entropy.}
We use Shannon entropy to quantify the diversity of expert activations at each layer. Intuitively, higher entropy indicates that many experts are invoked with comparable frequency, requiring more placements to ensure low communication overhead. In contrast, low-entropy distributions can be served well with fewer local experts.
\begin{lemma}[Entropy-Guided Coverage Lower Bound]
Let $\mathbf{p} = (p_1, \dots, p_E)$ be the activation distribution over $E$ experts in a layer, and let $H(\mathbf{p}) = -\sum_{e=1}^E p_e \log p_e$ denote its Shannon entropy. Then, for any $\delta \in (0,1)$, the number of experts needed to cover at least $(1-\delta)$ of the activation mass satisfies:
\[
k_\delta \geq 2^{H(\mathbf{p}) - \delta \log E}.
\]
\end{lemma}
\begin{proof}[Proof sketch.]
This follows from the asymptotic equipartition property~\cite{Cover}: the typical set covering $(1-\delta)$ of the probability mass has cardinality lower bounded by $2^{H(\mathbf{p}) - \delta \log E}$.
\end{proof}
This result formalizes why entropy is a meaningful proxy for expert count: higher entropy implies more uniform demand, which in turn necessitates broader expert placement. By allocating counts proportionally to entropy under memory constraints, Algorithm~\ref{alg:expert_count} achieves a principled trade-off between resource efficiency and coverage.

\subsubsection{\textbf{Expert-to-Server Assignment}}
Algorithm~\ref{alg:expert_assign} selects the specific experts to place on each server, given the per-layer expert counts $N_{n,l}$ from Algorithm~\ref{alg:expert_count}. Each server $n$ maintains a layer-wise preference list $\mathcal{P}_n^l$ over experts in $\mathcal{E}_l$, sorted by empirical activation frequency $f_n^l(e)$.
We initialize the expert set $\mathcal{A}_n$ by selecting the top-$N_{n,l}$ most frequently activated experts from $\mathcal{P}_n^l$ for each layer $l$.
To ensure full expert coverage across the system (i.e., each expert $e \in \mathcal{E}_l$ is placed on at least one server), we iteratively identify unassigned experts and reallocate them to servers holding redundant experts. For each unassigned expert $e \in \mathcal{U}_l$, we prioritize placement on servers with fewest duplicate assignments. Within each such server, we replace the least-frequently used duplicate with $e$.
This procedure balances expert coverage and preference while respecting the memory-aware expert counts computed earlier.

\begin{algorithm}[t]
\caption{Expert-to-Server Assignment}
\label{alg:expert_assign}
\SetAlgoVlined
\SetKwInOut{Input}{Input}
\SetKwInOut{Output}{Output}

\Input{Expert counts $N_{n,l}$ from Algorithm~\ref{alg:expert_count}; Per-server preference lists $\mathcal{P}_n^l$ sorted by activation frequency $f_n^l(e)$}
\Output{Expert assignment $\mathcal{A}_n$ for each server $n$}

\For{$n = 1$ to $N$}{
    \For{$l = 1$ to $L$}{
        $\mathcal{A}_n^l \leftarrow$ Top-$N_{n,l}$ experts from $\mathcal{P}_n^l$\;
    }
}

\For{$l = 1$ to $L$}{
    $\mathcal{U}_l \leftarrow \{ e \in \mathcal{E}_l \mid \sum_{n=1}^N \mathbf{1}\{e \in \mathcal{A}_n^l\} = 0 \}$\;
    \While{$\mathcal{U}_l \neq \emptyset$}{
        Sort servers $n$ by number of duplicates in $\mathcal{A}_n^l$ (ascending)\;
        \For{$n = 1$ to $N$}{
            $e \leftarrow$ most frequent unassigned expert in $\mathcal{U}_l$ (according to $f_n^l(e)$)\;
            \If{$e \notin \mathcal{A}_n^l$}{
                $e_{\text{rep}} \leftarrow$ least-used duplicate expert in $\mathcal{A}_n^l$\;
                Replace $e_{\text{rep}}$ with $e$ in $\mathcal{A}_n^l$\;
            }
        }
        Update $\mathcal{U}_l$\;
    }
}
\end{algorithm}

\textit{Theoretical Guarantee.}
Recall that our proxy objective in Eq.~\ref{eq:comm_aware_obj} aims to minimize the total expected number of remote expert invocations.
Since an expert incurs no remote communication if it is placed locally on the requesting server, minimizing this is equivalent to maximizing the local frequency mass of placed experts on each server $n$. That is, the term
\[
\sum_{l = 1}^L \sum_{e \in \mathcal{E}_l} f_n^l(e) \cdot (1 - \mathbf{1}_{\text{remote}}(n,e)) = \sum_{l=1}^L \sum_{e \in \mathcal{A}_n \cap \mathcal{E}_l} f_n^l(e)
\]
represents the communication-saving utility of the local assignment $\mathcal{A}_n$. This motivates defining the following local utility function and analyzing the greedy assignment procedure used in Algorithm~\ref{alg:expert_assign}:

\begin{thm}[Greedy Approximation for Local Expert Assignment]
\label{lem:greedy_approx}
Let $f_n^l(e) \in [0,1]$ denote the normalized activation frequency of expert $e \in \mathcal{E}_l$ on server $n$ at layer $l$. Let $B_n = \sum_{l=1}^L N_{n,l}$ be the total expert budget for server $n$ given by Algorithm~\ref{alg:expert_count}. Define the local utility:
\[
U_n(S) = \sum_{l=1}^L \sum_{e \in S \cap \mathcal{E}_l} f_n^l(e),
\]
for any subset $S \subseteq \bigcup_{l=1}^L \mathcal{E}_l$ of size at most $B_n$. Then, the greedy assignment $\mathcal{A}_n$ returned by Algorithm~\ref{alg:expert_assign} satisfies:
\[
U_n(\mathcal{A}_n) \ge (1 - 1/e) \cdot U_n(\mathcal{A}_n^*),
\]
where $\mathcal{A}_n^*$ is the optimal size-$B_n$ assignment maximizing $U_n(\cdot)$. Equivalently, the expected communication cost is:
\[
C_n = \sum_{l=1}^L \sum_{e \in \mathcal{E}_l \setminus \mathcal{A}_n} f_n^l(e)
\]
satisfies:
\[
C_n \le C_n^* + \Delta_n,
\]
where $C_n^*$ is the minimal achievable cost and $\Delta_n = U_n(\mathcal{A}_n^*) - U_n(\mathcal{A}_n)$ is the gap from greedy approximation.
\end{thm}

\begin{proof}[Proof sketch.]
The function $U_n(S)$ is monotonic and submodular as it aggregates independent, non-negative activation weights. The greedy algorithm that selects experts with the largest $f_n^l(e)$ values maximizes $U_n$ under a cardinality constraint, which yields a $(1 - 1/e)$ approximation. The cost equivalence follows directly by rewriting the proxy objective in Eq.~\eqref{eq:comm_aware_obj} in terms of unassigned experts.
\end{proof}

\subsubsection{\textbf{Expert Migration}}
\label{sec:mig}
To adapt to dynamic workloads, evolving activation patterns, or hardware fluctuations, the system periodically updates the expert placement strategy. At fixed intervals (e.g., every 5 minutes), the global scheduler re-runs the placement pipeline in Section~\ref{sec:our_placement} using the latest activation statistics, producing a candidate plan $\mathcal{P}'$.

Adopting a new placement introduces a migration cost due to transferring expert models across servers. We quantify it as:
\begin{align}
\label{eq:migration}
T_{\text{mig}}(\mathcal{P}, \mathcal{P}') = \sum_{n,g,e} \mathbb{I}[z_{n,g}^e \ne z_{n,g}^{e'}] \cdot \frac{m_e}{\texttt{speed}_{n,g}},
\end{align}
where $z_{n,g}^e$ and $z_{n,g}^{e'}$ are the expert placements before and after migration, $m_e$ is the expert size, and $\texttt{speed}_{n,g}$ denotes the I/O bandwidth of GPU $g$ on server $n$.
The system compares the total expected communication cost under the current plan $\mathcal{P}$ and the new plan $\mathcal{P}'$, using the proxy objective in Eq.~\ref{eq:comm_aware_obj}. The new plan is adopted only if the improvement in communication cost outweighs the migration overhead:
\begin{align}
C(\mathcal{P}') + T_{\text{mig}}(\mathcal{P}, \mathcal{P}') < C(\mathcal{P}),\label{eq:migration-condition}
\end{align}
where $C(\cdot)$ denotes the expected remote expert invocation cost under a given placement. This policy ensures that migrations occur only when they are beneficial in the long run.

        \section{Implementation and Experimental Results}
\label{sec:exp}

Based on the open-source code of MoE-infinity~\cite{moe-infinity}, we implement \name. Building upon the MoE-Infinity framework, we have developed a multi-server MoE inference system.
This flexible architecture allows us to accommodate various hardware configurations and optimize resource utilization to meet the demands of various deployment scenarios.

\begin{table*}
\begin{tabular}{clcccccccc}
\hline
\multicolumn{1}{l}{}                                                                                                    &                                      & \multicolumn{4}{c}{\textbf{DeepSeek-V2-Lite}}                                                    & \multicolumn{4}{c}{\textbf{Mixtral 8$\times$7B}}                            \\ \hline
\multicolumn{1}{c|}{\textbf{Dataset}}                                                                                   & \multicolumn{1}{c|}{\textbf{Method}} & \textbf{Server1} & \textbf{Server2} & \textbf{Server3} & \multicolumn{1}{c|}{\textbf{Total Avg}} & \textbf{Server1} & \textbf{Server2} & \textbf{Server3} & \textbf{Total Avg} \\ \hline
\multicolumn{1}{c|}{\multirow{5}{*}{\textbf{\begin{tabular}[c]{@{}c@{}}BigBench\\ (10s Poisson\\ arrival)\end{tabular}}}}  & \multicolumn{1}{l|}{Uniform}         & 48.55            & 17.28            & 9.43             & \multicolumn{1}{c|}{21.66}              & 6.44             & 6.04             & 3.44             & 5.08               \\
\multicolumn{1}{c|}{}                                                                                                   & \multicolumn{1}{l|}{Redundance}      & 25.88            & 13.77            & 4.79             & \multicolumn{1}{c|}{13.08}              & 5.96             & 6.16             & 3.07             & 4.85               \\
\multicolumn{1}{c|}{}                                                                                                   & \multicolumn{1}{l|}{SmartMoE}        & 40.28            & 16.43            & 9.51             & \multicolumn{1}{c|}{19.39}              & 3.76             & 1.59             & 2.03             & 2.30               \\
\multicolumn{1}{c|}{}                                                                                                   & \multicolumn{1}{l|}{EPLB}            & \underline{23.93}      & \underline{7.08}       & \underline{3.12}       & \multicolumn{1}{c|}{\underline{9.56}}         & \underline{3.60}       & \textbf{1.53}    & \underline{1.82}       & \underline{2.16}         \\
\multicolumn{1}{c|}{}                                                                                                   & \multicolumn{1}{l|}{Ours}            & \textbf{14.67}   & \textbf{5.85}    & \textbf{2.49}    & \multicolumn{1}{c|}{\textbf{6.63}}      & \textbf{3.52}    & \underline{1.57}       & \textbf{1.67}    & \textbf{2.09}      \\ \hline
\multicolumn{1}{c|}{\multirow{5}{*}{\textbf{\begin{tabular}[c]{@{}c@{}}MultiData\\ (20s Poisson\\ arrival)\end{tabular}}}} & \multicolumn{1}{l|}{Uniform}         & 24.85            & 30.63            & 36.48            & \multicolumn{1}{c|}{30.65}              & 9.20             & 10.26            & 13.31            & 10.92              \\
\multicolumn{1}{c|}{}                                                                                                   & \multicolumn{1}{l|}{Redundance}      & \underline{15.96}      & 24.76            & 12.36            & \multicolumn{1}{c|}{17.70}              & \textbf{8.28}    & 10.32            & 11.08            & 9.90               \\
\multicolumn{1}{c|}{}                                                                                                   & \multicolumn{1}{l|}{SmartMoE}        & 23.68            & 30.41            & 35.21            & \multicolumn{1}{c|}{29.77}              & 9.82             & 11.89            & 13.84            & 11.85              \\
\multicolumn{1}{c|}{}                                                                                                   & \multicolumn{1}{l|}{EPLB}            & 16.36            & \underline{19.07}      & \underline{9.35}       & \multicolumn{1}{c|}{\underline{14.93}}        & 8.69             & \textbf{9.52}    & \underline{9.89}       & \underline{9.36}         \\
\multicolumn{1}{c|}{}                                                                                                   & \multicolumn{1}{l|}{Ours}            & \textbf{13.96}   & \textbf{16.72}   & \textbf{9.20}    & \multicolumn{1}{c|}{\textbf{13.29}}     & \underline{8.30}       & \underline{9.80}       & \textbf{9.41}    & \textbf{9.17}      \\ \hline
\end{tabular}
\caption{Serve latency (seconds) comparison between DeepSeek-V2-Lite and Mixtral 8$\times$7B models across different methods on BigBench (10s Poisson arrival) and MultiData (20s Poisson arrival) datasets. Results are reported for three servers and their total average. Bold values represent the best serving latency, and underlined values represent the second best. \name achieves the lowest average latency, particularly on models with a larger number of experts such as DeepSeek-V2-Lite. }
\label{tab:all_latency_results}
\end{table*}

\subsection{Experimental Setup}
\para{Hardware.} Our testbed comprised a machine with four A100 (40GB) GPUs and 256GB RAM, configured to simulate three edge servers with GPU allocations of 1, 1, and 2 respectively. Inter-server communication was established through a Docker network, with bandwidth limits constrained to 500 Mbps using Traffic Control (tc) in the Linux kernel.

\para{MoE Model.} Our experiments employ two MoE architectures: (1) Mixtral-8x7B~\cite{jiang2024mixtral} (32 layers with 2 active experts out of 8 per layer), whose complete parameter set exceeds single-GPU memory capacity, and (2) Deepseek-V2-Lite~\cite{liu2024deepseekv2} (26 layers selecting 8 active experts from 64 per layer, totaling 1,664 possible expert combinations), which presents greater routing complexity while being more memory-efficient. To simulate realistic edge deployment scenarios, we artificially constrain GPU memory allocation to 70\% of total capacity for Mixtral-8x7B and 30\% for Deepseek-V2-Lite, reflecting their respective memory requirements and operational constraints.

\para{Dataset.} Our experimental evaluation incorporated four widely adopted benchmark datasets: BIG-bench~\cite{srivastava2023beyond}, MMLU-Pro~\cite{wang2024mmlu}, WikiText~\cite{merity2016pointer}, and Tako~\cite{li2023taco}. 
\begin{itemize}[wide, labelindent=0pt]
\item BIG-bench: The Beyond the Imitation Game Benchmark (BIG-bench) is a collaborative evaluation framework comprising over 200 diverse tasks. We constrained the model's output length to match the answer length specified in each dataset.
\item MMLU-Pro: MMLU-Pro has more than 12K questions across 14 domains. We format inputs as question-choice pairs and constrain output length to the correct answer's length.
\item WikiText: The WikiText language modeling dataset is curated from Wikipedia's verified Featured articles. We constrain the model's maximum output length to 20 tokens.
\item Tako: While evaluating on TACO (a benchmark for code generation tasks), we use its test set problems as model inputs while constraining the maximum output length to 20 tokens.
\end{itemize}

\para{Server configurations.} We evaluated two server setups: a specialized setup where each server handled distinct BIG-bench tasks (abstract narrative, arithmetic reasoning, ASCII recognition), and a heterogeneous setup distributing MMLU-Pro, WikiText and Tako across three servers respectively.  

\para{Baselines.}
We conducted comparisons between our proposed method and the following baseline approaches:
\begin{itemize}[wide, labelindent=0pt]
\item Uniform. Experts are uniformly distributed across all available devices. For instance, in our experimental setup with 4 GPUs, each device hosts 2 random experts per layer for the Mixtral model and 16 experts per layer for Deepseek-V2-Lite. This placement method aligns with the expert parallelism implementation in the Megatron-LM~\cite{shoeybi2019megatron} project.
\item Redundance. We propose this heuristic that enables expert duplication across servers by randomly distributing experts up to each device's capacity. When the total available GPU memory in the system surpasses the model's memory demands, such redundancy  improves memory utilization compared to Uniform placement.
\item SmartMoE~\cite{zhai2023smartmoe}. While SmartMoE's core focus is training optimization, our framework specifically implements SmartMoE's placement module, which strategically distributes experts across GPUs based on real-time workload. The placement algorithm balances workload distribution to maintain uniform expert allocation across devices.
\item Expert Parallelism Load Balancer (EPLB)~\cite{deepseek2025eplb}. EPLB is a redundant expert strategy in DeepSeek-V3~\cite{liu2024deepseek} that duplicates high-load experts and heuristically distributes them to balance GPU workloads. Since the open-source implementation only supports homogeneous systems, we re-implement the algorithm within our framework to accommodate heterogeneous system configurations.
\end{itemize}

\para{System Parameter Settings.} The system evaluates potential expert migrations every 5 minutes. We employ historical communication and computation time of expert execution as estimation metrics. Specifically, the average values of all executions between the last placement change and the current moment are used as reference for calculating the cost of new placement. Computation time is a load-dependent variable which is updated at 30-second intervals.

\subsection{System Performance}
\para{Inference Latency.}
As shown in Tab.~\ref{tab:all_latency_results}, robustness across datasets and server configurations highlights practical superiority of \name.
Uniform and SmartMoE achieve limited success due to inefficient resource utilization, while Redundance, though introducing expert duplication, suffers from higher latency due to suboptimal implementation.
On the BigBench dataset, \name reduces the average latency by 30.6\% compared to EPLB, the second-best method for DeepSeek-V2-Lite. For Mixtral, it achieves a slight but consistent improvement over EPLB. On the MultiData dataset, \name demonstrates a 10.9\% reduction in latency compared to EPLB for DeepSeek-V2-Lite and a 2.0\% improvement for Mixtral.
Since EPLB is a load-balancing algorithm that does not account for cross-machine communication reduction, it underperforms compared to \name. For the Mixtral model, where the number of experts is relatively small and server capacity is limited, the room for optimization is narrower, resulting in smaller performance gains for \name.

\para{Local Compute Ratio.}
To assess the impact of cross-machine communication on inference latency, we measured the local computation ratio as an indicator of communication overhead. Except for the Uniform and Redundance baselines, all methods adopt \name's migration strategy but differ in their expert placement strategies. As shown in Fig.~\ref{fig:request_radio}, \name consistently achieves higher local computation ratios. Notably, the advantage becomes more pronounced following explicit migration events triggered at 5 minutes (BigBench) and 10 minutes (MultiData).

\begin{figure}[htb]
    \centering
     \includegraphics[width=0.49\textwidth]{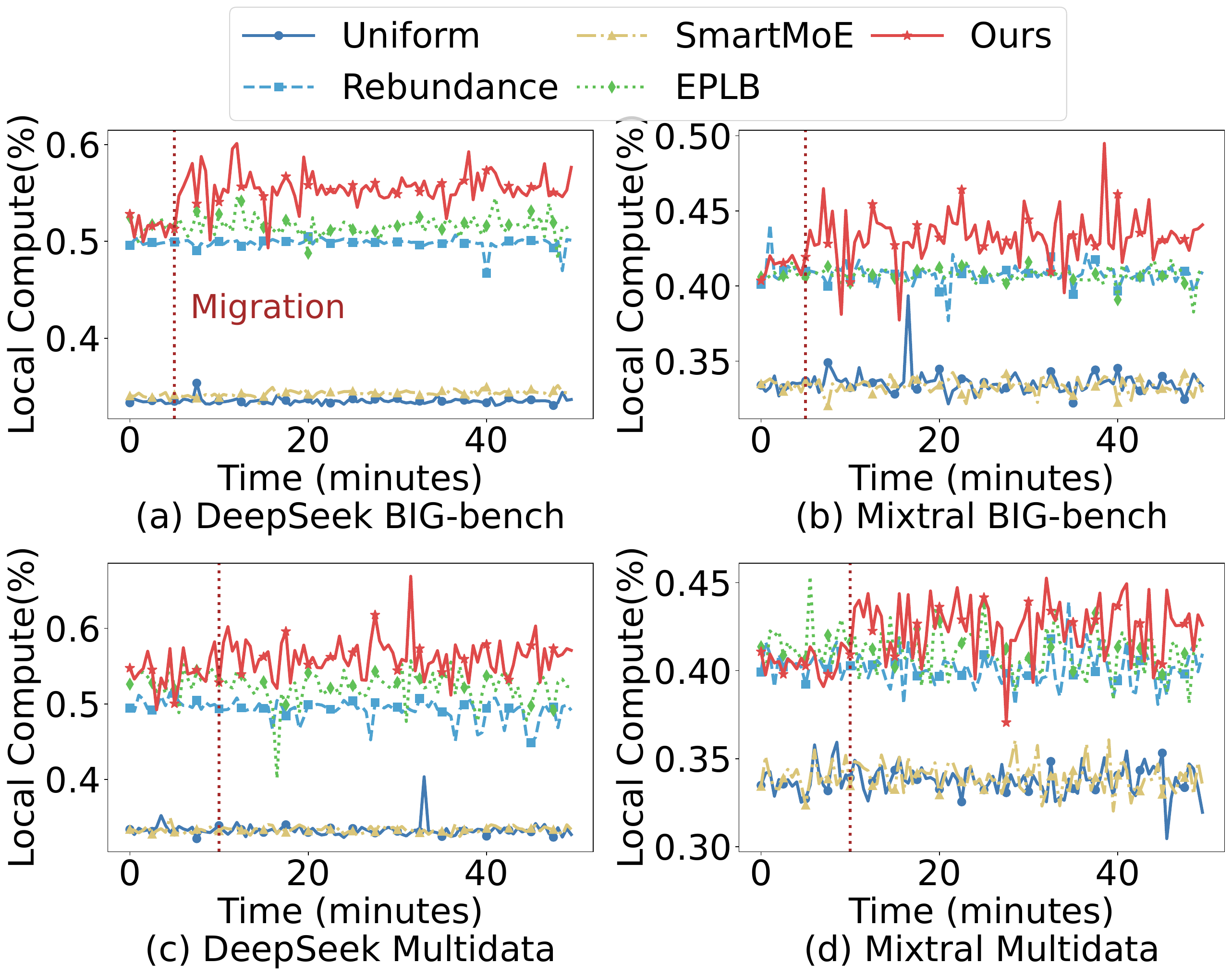}
    \caption{Evolution of local compute ratio over inference runtime for different methods across model-dataset configurations.}
    \label{fig:request_radio}
\end{figure}

\para{Effectiveness of Migration.}
We evaluated migration effectiveness by comparing adaptive and static systems using Deepseek-V2-Lite (Fig.~\ref{fig:migration_analysis}), from 200 MultiData to 200 BIG-bench requests per server. Both systems showed identical initial performance, but the migration-enabled approach (``w'' in Figure~\ref{fig:sub1}) achieved significantly higher local computation after the first migration. 
We now explain the three migrations shown in Figure~\ref{fig:sub1}. The first one is triggered because the system detects a divergence between the accumulated inference data and the dataset used for the initial placement, activating the condition in Eq.\eqref{eq:migration-condition}. The second migration follows an inference data shift (see Data Change in Figure~\ref{fig:sub2}). However, since the cost $C(\cdot)$ in Eq.\eqref{eq:migration-condition} calculated at this time is based on the accumulated inference data mixing the old and new data, a third migration occurs later as more new data accumulates. 

The migration version achieved a 10\% reduction in average latency (7.48 to 6.73) of all requests compared to the non-migration.

\begin{figure}[htb]
    \centering
    \begin{subfigure}[b]{0.43\textwidth}
        \includegraphics[width=\textwidth]{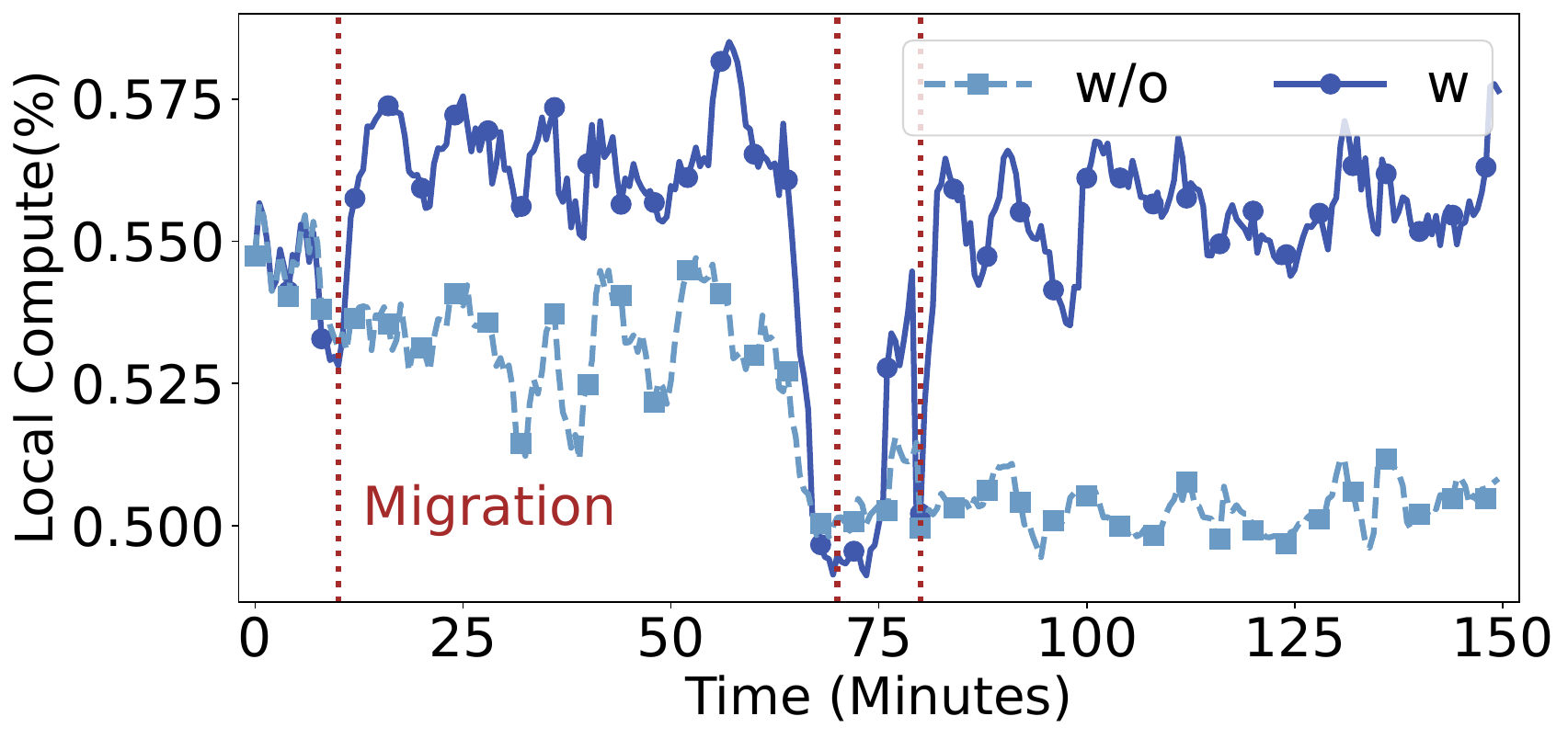}
        \caption{Ratio of computations successfully completed locally without cross-server communication.}
        \label{fig:sub1}
    \end{subfigure}
    
    \vspace{0.1cm}
    
    \begin{subfigure}[b]{0.43\textwidth}
        \includegraphics[width=\textwidth]{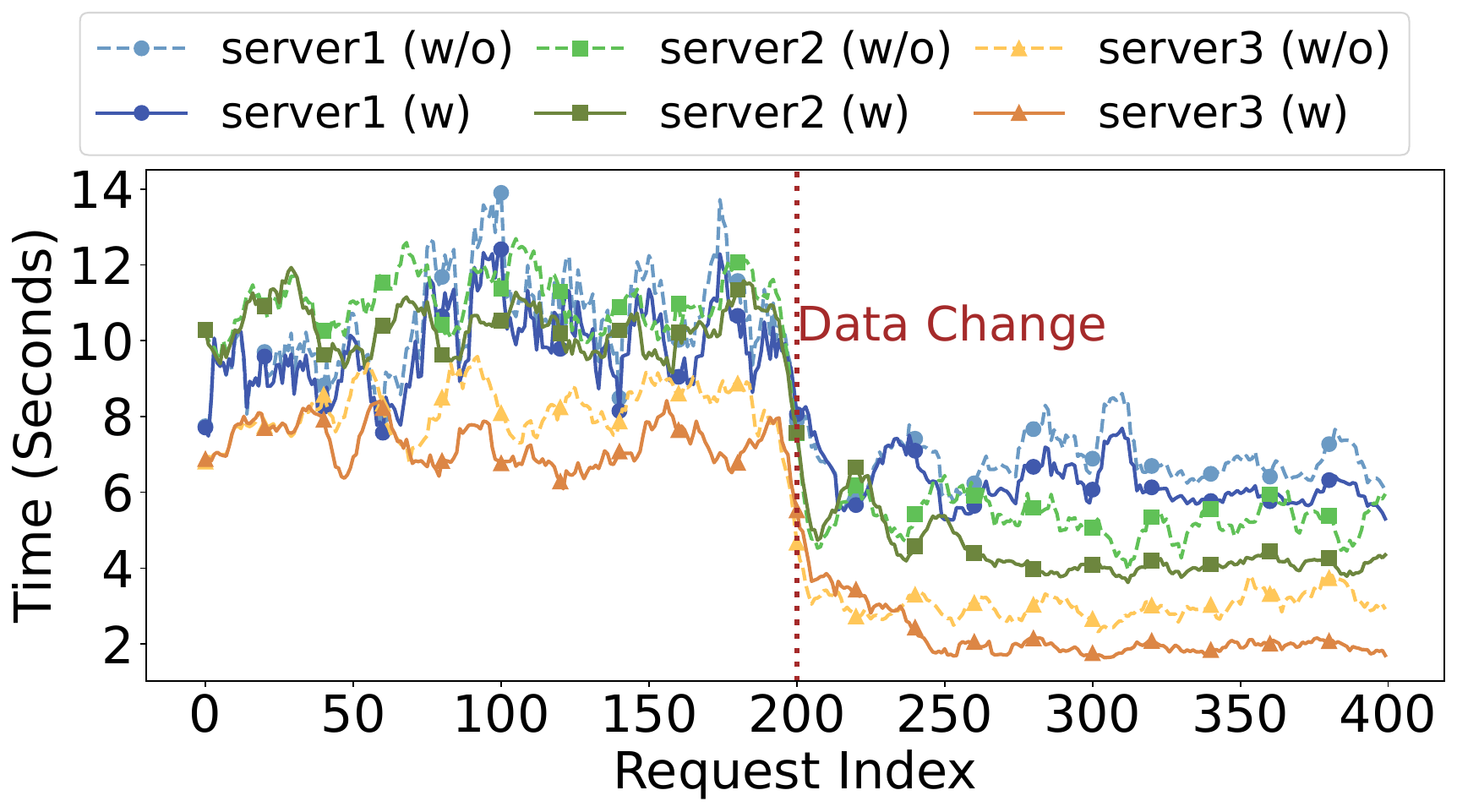}
        \caption{Latency comparison of migration across servers.}
        \label{fig:sub2}
    \end{subfigure}
    
    \caption{Comprehensive evaluation of migration efficiency. Subfigure (a) shows the improvement in local task computation enabled by the proposed migration strategy, while subfigure (b) evaluates the latency impact when migration becomes necessary across different servers.}
    \label{fig:migration_analysis}
\end{figure}

\subsection{Large-scale Simulation}

\para{Objective.}
To validate the scalability and performance of our proposed system, we develop an event-driven simulator with the following objectives: (1) demonstrate performance improvement at scale by varying GPU counts from 4 to 256; (2) quantify the impact of communication bandwidth (configured using the Linux tc tool across a range of 100 Mbps to 1000 Mbps) on system inference run time.

\para{Simulation Setup.}
The simulator's primary function orchestrates the inference process by: selecting the next layer for processing; updating server timestamps; recording completion times; computing inference latency. Through incremental time accumulation, it accurately models the temporal progression of the entire system.
The simulator contains several key components: (1) Prompt Sequence Generator: We integrate a Poisson arrival time sequence with operational data collected from \name, the dataset contains both expert selection patterns and token processing volumes. (2) Prompt Routing Generator: To simulate data flow under varying expert placement configurations. (3) Communication \& Computation Time Consumption Estimator: We develop a linear model to predict processing time per token batch. (4) Time Stamp Calculator. This module captures both communication and computation events, timestamping each operation. (5) System Timeline Scheduler. The scheduler carefully interleaves communication and computation events according to our prescribed design.

\para{Performance Evaluation.}
As shown in Fig.~\ref{fig:simulation_sub1}, as GPU volume increases, the average time consumption per prompt decreases by 9\% (Poisson, 15s arrival) to 19\% (Poisson, 8s arrival); furthermore, our simulations demonstrate that performance improvements are more pronounced when input prompts are more intensive (Poisson 8s arrival compared with Poisson 15s arrival); Fig.~\ref{fig:simulation_sub2} demonstrates that higher bandwidth yields a substantial reduction in average processing time, achieving over 55\% improvement in the 4-GPU configuration; however, this benefit diminishes with GPU scaling, declining to just 35\% for the 256-GPU configuration.

\begin{figure}[htb]
    \centering
    \begin{subfigure}[b]{0.24\textwidth} 
        \includegraphics[width=\textwidth]{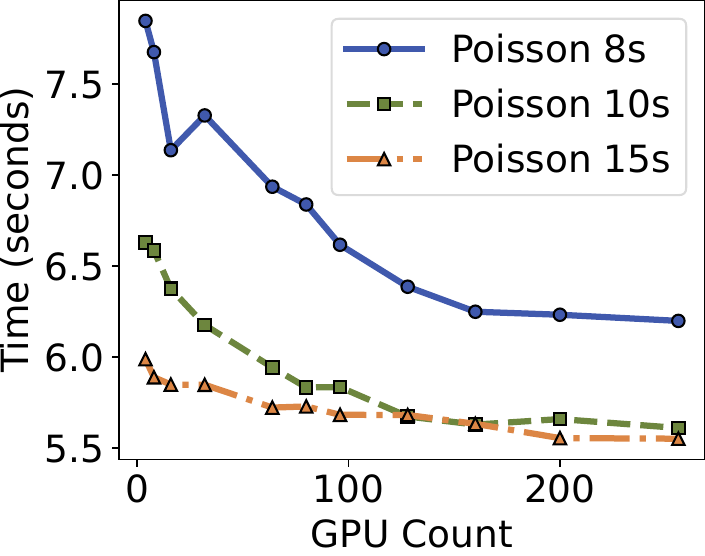}
        \caption{Time cost of different GPU counts and queuing latency.}
        \label{fig:simulation_sub1}
    \end{subfigure}
    \begin{subfigure}[b]{0.24\textwidth}
        \includegraphics[width=\textwidth]{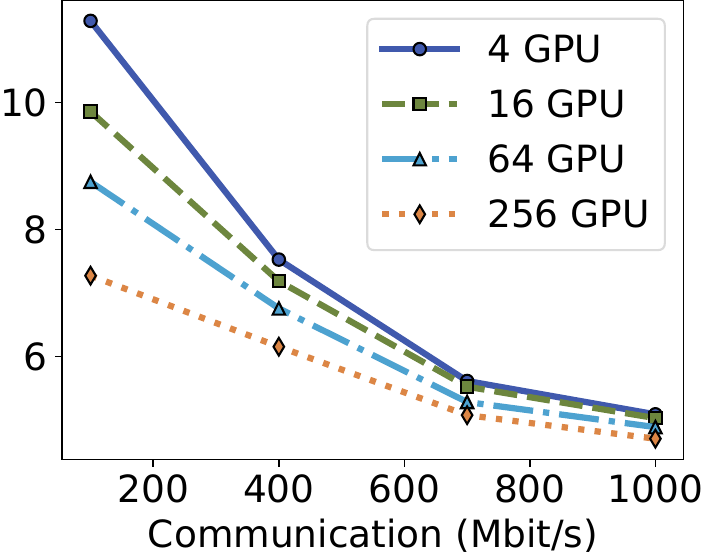}
        \caption{The influence of bandwidth limitation.}
        \label{fig:simulation_sub2}
    \end{subfigure}
    
    \caption{Simulation results for system scalability verification. Subfigure (a) shows potential improvement of system latency with GPU volume increase, while (b) indicates the bandwidth limitation has a significant impact on system performance.}
    \label{fig:simulation}
\end{figure}

        \section{Related Work}
\label{sec:related_1}
Several edge serving works~\cite{li2022joint, li2023task, li2023tapfinger} have explored efficient strategies to achieve low-latency application services, yet they fail to account for the unique architectural characteristics and parallelization patterns inherent in MoE models.
Recent efforts to optimize Mixture-of-Experts (MoE) models span both training and inference phases, targeting challenges such as expert placement, memory efficiency, and communication overhead. For training scenarios, FasterMoE~\cite{he2022fastermoe} optimizes training by trading off hidden-state transfers against expert migrations, while SmartMoE~\cite{zhai2023smartmoe}, FlexMoE~\cite{nie2023flexmoe} and Prophet~\cite{wang2023prophet} develop GPU load balancing strategies. Lazarus~\cite{wu2024lazarus} further considers system-level scalability. However, these training-oriented approaches fail to address the heterogeneous communication and computational capabilities inherent in multi-server inference environments.
For inference on memory-constrained devices, standalone systems have explored multi-level memory hierarchies and expert prefetching strategies. Methods include storing excess experts in CPU memory (as in Pre-gated MoE~\cite{hwang2024pre}, which introduces pre-gating for parallel loading at potential accuracy cost), predicting expert activations to prefetch parameters (Lina~\cite{li2023accelerating}, MoE-Infinity~\cite{xue2024moe}, EdgeMoE~\cite{yi2023edgemoe}, AdapMoE~\cite{zhong2024adapmoe}, SwapMoE~\cite{kong2024swapmoe}, SIDA~\cite{du2024sida}), and computing experts directly on CPU (Fiddler~\cite{kamahori2024fiddler}). While effective for individual resource-constrained devices, these solutions neglect cross-server resource coordination. In contrast, \name systematically coordinates distributed resources through optimized collaboration.

More recently, inference-time expert placement has been explored in systems like Moetuner~\cite{go2025moetuner} and DeepSeek-V3’s EPLB~\cite{deepseek2025eplb}. These methods use historical routing or expert redundancy to balance loads, but are generally designed for cluster-scale datacenter deployment, without accounting for edge-specific constraints like heterogeneous hardware or limited bandwidth. Complementary to placement strategies, communication-aware MoE systems, such as PipeMoE~\cite{shi2023pipemoe}, ScheMoE~\cite{shi2024schemoe}, and TUTEL~\cite{hwang2023tutel}, aim to improve system efficiency through pipelining, hybrid parallelism, or fine-grained communication scheduling. While primarily designed for training, their techniques offer complementary benefits for distributed inference. In contrast to these prior efforts, \name targets collaborative inference across memory-limited edge servers. It introduces activation-aware placement and lightweight migration strategies tailored to heterogeneous devices and variable workloads—an underexplored yet increasingly important setting for scalable MoE deployment.

        \section{Conclusion}
\label{sec:conclusion}
In this paper, we present \name, a novel system designed for efficient MoE inference in edge computing environments. Unlike existing solutions tailored for cloud-scale clusters, \name is specifically architected to address the unique challenges of edge deployments, such as limited device memory, heterogeneous hardware, and inefficient network conditions. A few recent works on edge MoE propose memory-efficient strategies, targeting either a single device or homogeneous resource configurations. To this end, \name introduces adaptive expert placement and dynamic migration mechanisms that jointly optimize resource utilization and inference efficiency under workload variability. Extensive testbed experiments and simulation studies demonstrate the system’s effectiveness, adaptability to dynamic conditions, and scalability in real-world edge scenarios.
	\clearpage
	\bibliographystyle{IEEEtran}
	\bibliography{main.bib}

\end{document}